\documentclass[12pt]{amsart}
\usepackage{amssymb, amsmath, amsthm}
\usepackage{mathrsfs}
\usepackage{hyperref} 
\usepackage{natbib} 
\usepackage[foot]{amsaddr} 
\newcommand{\RR}{{\mathbb{R}}}
\newcommand{\PP}{{\mathbb{P}}}
\newcommand{\A}{\mathscr A}
\newcommand{\ov}{\mathbf{e}}
\newcommand{\ox}{\mathbf{x}}
\newcommand{\jv}{\mathbf{j}}

\theoremstyle{plain}
\newtheorem{theorem}{Theorem}[section]
\newtheorem{proposition}[theorem]{Proposition}

\theoremstyle{definition}
\newtheorem{remark}[theorem]{Remark}
\begin{document}
\title{Future Exchange Rates and Siegel's Paradox}
\let\thefootnote\relax\footnotetext{To appear in Global Finance Journal. \url{https://doi.org/10.1016/j.gfj.2018.04.007}}
\author{Keivan Mallahi-Karai}
\address{Authors' affiliations and E-mails:}
\thanks{Keivan Mallahi-Karai, Jacobs University, Campus Ring I, 28759 Bremen, Germany. \texttt{k.mallahikarai@jacobs-university.de}}
\author{Pedram Safari}
\thanks{Pedram Safari (corresponding author), Department of Mathematics and Institute for Quantitative Social Science, Harvard University, Cambridge, MA 02138, USA. \texttt{safari@fas.harvard.edu}}
\thanks{\copyright\, 2018. This manuscript version is made available under the CC-BY-NC-ND 4.0 license 
\\ \url{http://creativecommons.org/licenses/by-nc-nd/4.0/}}
\date{} 
\maketitle
\begin{abstract}
Siegel's paradox is a fundamental question in international trade about exchange rates for futures contracts and has puzzled many scholars for over forty years. 
The unorthodox approach presented in this article leads to an {\it arbitrage-free} solution which is invariant under currency re-denominations and is symmetric, as explained. 
We will also give a complete classification of all such aggregators in the general case. The formula obtained in this setting therefore describes all the negotiated no-arbitrage forward exchange rates in terms of a reciprocity function. 
\end{abstract}
%
%
\section{Introduction}
\par
Siegel's paradox, discussed by \cite{Siegel72}, is a discount bias 
in future exchange rates and is often discussed in connection with 
\cite{Nalebuff} puzzle in the context of expected values.
It can briefly be described in simple terms as follows. Let us assume that there are two 
possible states of the world $\omega_1$ and $\omega_2$ at a specific time 
$T$ in the future, both having an equal chance of 50\% of occurring, where 
the exchange rate of Euros to US dollars is expected to be $e_1$ and $e_2,$ 
respectively. It appears that an investor who wants to exchange Euros to 
US dollars at time $T$ would consider the expected value 
$\frac12(e_1+e_2)$ as the spot price for a futures contract at time 0. 
Meanwhile, for their trading counterpart, who wants to perform the reciprocal exchange of 
US dollars to Euros, this price would be $\frac12(e_1^{-1}+e_2^{-1})$ in their own currency, which 
is an obvious disagreement. This is the content of Siegel's paradox --- 
that these (risk-neutral) investors cannot both opt for the arithmetic 
mean as the spot price for a futures contract. In other words, the source 
of the paradox lies in the fact that the arithmetic and harmonic means do 
not coincide.
\par
It should not be surprising that this problem is of great significance in international trade. 
The Bank for International Settlements estimates that the {\it daily} turnover in foreign exchange markets 
far exceeds \$1 trillion (\cite{Rogoff}).
There have been numerous theoretical attempts at understanding this paradox, as well as enormous amount of empirical testing. To begin with, \cite{Siegel72}  
seems to have viewed it not as a paradox, but as a relationship 
between foreign exchange prices and interest rates under risk neutrality.  
However, in the absence of interest rates, the paradox 
will not fade away. He suggested that the investors in the "small country" 
choose a biased estimator, as their own forward rate, which lies between 
the geometric and harmonic means. 
\cite{Roper75} suggested that the paradox may be resolved if the investors take their forward profits in the 
foreign currency. However, if the investors wish to take their profits in 
their own domestic currencies, they have to be risk-averse (\cite{Beenstock85}).
A nice overview of Siegel's paradox could be found in \cite{Edlin02}.
None of the proposed solutions, to our knowledge, suggests a specific forward 
exchange rate that is acceptable to risk-neutral investors in both currencies.
\par
In their classic text on international macroeconomics, \cite{Rogoff} 
dedicate an entire section to Siegel's paradox and its ramifications. They analyze the empirical tests for prediction bias in forward rates in detail and develop a stochastic monetary model to understand forward exchange pricing. They explain that in their model, if (relative) {\it purchasing power parity} holds, the expected real returns will be zero for one party if it is zero for the other, but the argument fails when the PPP fails for any reason. They also suggest that the forward exchange rate in the equilibrium ends up to be a negotiated rate between the expected value of the future spot rate from one investor's perspective and that of their counterpart --- 
in other words, if ${\mathcal E}_{T}$ is the future spot rate at time T in the future, the negotiated rate will be between $E({\mathcal E}_{T})$ and $1/{E(1/{\mathcal{E}_T})},$ which is consistent with our findings.
\par
In this paper, we take a non-traditional approach and seek an aggregator that is 
arbitrage-free, symmetric and invariant under redenominations. 
These conditions will be defined precisely in the next section. 
We will show that under these natural axioms, 
the only possible aggregator for Siegel's paradox will be the geometric mean. Note that 
for any pair of positive numbers, the geometric mean always lies between the 
arithmetic and harmonic means. 
We will go even further and give a complete classification of the aggregators satisfying the generalized axioms in any dimension. The main result (Theorem \ref{mainn}) gives a formula 
for these aggregators in terms of the geometric mean and a reciprocity function. 
\par
Our approach not only provides an unbiased common ground to the exchange problem in Siegel's paradox, 
but might eventually even shed a new light in
understanding other questions in future discount rates, such as in the
Weitzman-Gollier puzzle (\cite{Weitzman-Gollier}).
\section{The Model}
\par
In this model, we will assume that there are two possible states of the world $ \omega_1$ and $ \omega_2$ at a fixed time $T$ in the future. 
We consider the possible values of foreign exchange rate (we use European and American currency) EUR/USD at time $T$ and assume that it can attain values 
$$R_{_{\rm EUR/USD}}(\omega_1)=e_1, \, R_{_{\rm EUR/USD}}(\omega_2)=e_2.$$
\par 
The reciprocal exchange rates for USD in terms of Euros will be $1/e_1,1/e_2,$ that is, 
$$R_{_{\rm USD/EUR}}(\omega_1)=e_1^{-1}, \, R_{_{\rm USD/EUR}}(\omega_2)=e_2^{-1}.$$
\par
Siegel's paradox shows that investors who want to exchange Euros to US dollars at time $T$ and the ones
who want to exchange US dollars to Euros at time $T$ cannot both use the arithmetic mean as the spot price
for the futures contract at time $0$.  
\par
Our point of view is that the arithmetic mean (or the expected value with respect to probability
distribution $(1/2,1/2)$) is not the right {\it ``aggregator''} to be used in this context.  
We will consider three natural axioms that any reasonable aggregator must satisfy.
Underlying each axiom is what may be viewed as an {\it invariance principle} stating the aggregate price 
must be preserved under certain transformations of the given data. Surprisingly, these three axioms are 
strong enough to determine a unique aggregator. Let us denote the aggregator by $\A(e_1,e_2)$. 
Our axioms are as follows.
\par
\begin{itemize}
\item {\bf Symmetry}. We will assume that there is no particular order on the two possible states of the world, hence the aggregator must be indifferent to the order in which the rates are listed.  
This implies that 
$$\A(e_1,e_2)= \A(e_2,e_1).$$
\item {\bf Re-denomination}. Upon redenomination of currencies, say, replacing a Euro by $100$ Euro cents, the exchange rates have to be adjusted accordingly, here by a factor of $100$. More generally, if we denote the redenomination factor by $ \lambda$, the following identity must hold:
$$\A( \lambda e_1, \lambda e_2)= \lambda \cdot \A(e_1, e_2).$$
\item {\bf Reciprocity}. This axiom explicitly precludes what happens in Siegel's paradox. It can also naturally be viewed as a {\it no-arbitrage} constraint for the exchange rates, stipulating that the aggregator prevents the possibility of making risk-less money by exchanging Euros to Dollars and then back into Euros (assuming, obviously, that there are no transaction costs). This condition can be expressed as
$$\A(e_1^{-1},e_2^{-1})=\A(e_1,e_2)^{-1}.$$
\end{itemize}
\section{The Geometric Mean}
\par
We will identify the set $E$ of, say, EUR/USD (and USD/EUR) exchange rates with the set of positive real numbers. An aggregator is a function $\A: E \times E \to E$, which aggregates 
the two possible exchange rates $e_1, e_2$ into one deterministic rate $\A(e_1, e_2)$. 
The following theorem gives a complete characterization of the aggregators $\A$ that satisfy the Symmetry, Redenomination and Reciprocity axioms. 
\par
\begin{theorem}\label{main2}
Let $\A: E \times E \to E$ be an aggregator, which satisfies the following axioms:
\begin{enumerate}
\item[\bf{A1.}] {\it Symmetry:} For all $e_1,e_2 \in E$, we have $\A(e_1,e_2)=\A(e_2,e_1).$
\item[\bf{A2.}]  {\it Redenomination:} For all $e_1,e_2 \in E$ and every $\lambda>0$, we have
$$\A( \lambda e_1, \lambda e_2)= \lambda \cdot \A(e_1. e_2).$$
\item[\bf{A3.}]  {\it Reciprocity:} For all $e_1,e_2 \in E$, we have 
$$\A(e_1^{-1},e_2^{-1})=\A(e_1,e_2)^{-1}.$$
\end{enumerate}
Then,
$$\A(e_1, e_2)= \sqrt{e_1 e_2}$$
for all $e_1,e_2 \in E.$ 
\end{theorem}
\par
\begin{proof}
It will be expedient to make a logarithmic change of coordinates. In order to do this, we set 
$\alpha(x,y)=\log(\A(e^x, e^y))$, where $\log$ denotes the natural logarithm. 
It is easy to verify that {\bf{A1}}--{\bf{A3}} correspond to the following properties of the new function $\alpha$, defined for all real numbers. 
\begin{enumerate}
\item[\bf{\^A1.}] For all $x,y \in \RR,$ we have $ \alpha(x,y)= \alpha(y,x).$
\item[\bf{\^A2.}] For all $x,y \in \RR,$ and all $ \lambda >0,$ we have $\alpha(x+ \lambda, y+ \lambda)= \alpha(x,y)+ \lambda.$
\item[\bf{\^A3.}] For all $x,y \in \RR,$ we have $\alpha(-x,-y)=-\alpha(x,y).$
\end{enumerate}
Now, we consider the difference
$$h(x,y)=\alpha(x,y)-\frac{x+y}{2}.$$
It is clear that {\bf {\^A1}} and {\bf {\^A3}} imply  
\begin{equation}\label{SymmOdd}
h(x,y)=h(y,x) \qquad , \qquad h(-x,-y)=-h(x,y).
\end{equation}
From {\bf{\^A2}} it follows that for all real values of $\lambda,$ we have
\begin{equation}\label{invariance}
h(x+\lambda, y+\lambda)= h(x,y).
\end{equation}
Applying \eqref{invariance} to $\lambda=-(x+y),$ we get
$$h(x,y)= h(-y,-x)=-h(y,x)=-h(x,y),$$
which shows that $h(x,y)=0.$ Hence $\alpha(x,y)=\dfrac{x+y}{2}.$ Now, by translating this back to the original variables, we obtain
$$\A(e_1,e_2)=\exp(\alpha(\log e_1, \log e_2))= 
\exp\left(\frac{\log e_1+ \log e_2}{2}\right)=\sqrt{e_1e_2}.$$
\end{proof}
\par
\begin{remark}
Using the geometric mean as an aggregator has the following curious interpretation. It is clear that from the point of view of an investor who is interested in a Euro to USD exchange in the future, the value of the geometric mean aggregator $\A(e_1,e_2)=\sqrt{e_1e_2}$ can be seen as the expected value with respect to a new probability distribution $\PP[\omega_1]=p, \PP[\omega_2]=1-p,$ hence
$$\sqrt{e_1e_2}= p e_1+ (1-p)e_2.$$
The new probability measure always gives a larger weight to the state of the world in which $e_i$ is smaller.
Similarly, the investor who wants to exchange US dollars to Euros in the future can also view the value
of the aggregator as the expected value with respect to some probability measure. What is interesting is that
this probability measure assigns $p$ to $\omega_2$ and $1-p$ to $\omega_1$. The reason for this is:
$$(1-p)\frac{1}{e_1}+ p\frac{1}{e_2}= \frac{pe_1+(1-p)e_2}{e_1e_2}=\frac{\sqrt{e_1e_2}}{e_1e_2}=\frac{1}{\sqrt{e_1e_2}}.$$ 
\end{remark}
\begin{remark}
The proof of theorem \ref{main2} could be re-interpreted in the language of group actions. To see this more clearly, let us consider the more general case of aggregators with $n$ variables $\A(e_1, \dots , e_n).$
One can define analogously an $n$-variable function $h,$ which now should be invariant under all the permutations $P_{\sigma}$ of the coordinates, as well as the translation group $T_{\lambda}$ along the vector $\ov_\lambda= (\lambda,\lambda, \dots, \lambda).$ Moreover, $h$ has to be odd with respect to the reflection $R$ across the origin. It is easy to see that the group $\Gamma$ generated by the transformations $P_{\sigma}, T_{\lambda}$ and $R$ is a finite extension of a group isomorphic to the additive group of real numbers consisting of $T_{\lambda}.$ In fact, any element of this group can be represented by 
$\gamma (x_1, \dots, x_n)= \epsilon(\gamma) (x_{\sigma(1)}, \dots , x_{ \sigma(n)})+ \ov_\lambda,$
where $\epsilon(\gamma)=\pm 1,$ $\sigma$ is a permutation, and $\lambda \in \RR.$ Alternatively, one can represent this group by matrices of the form
$$\begin{pmatrix}
\pm P_{\sigma}  &  \lambda \jv  \\
 0 & 1   \\
\end{pmatrix},$$
where $P_{\sigma}$ is the permutation matrix associated to $\sigma$ and $\jv$ is the column vector whose entries are all equal to $1.$ 
One can show that $\epsilon: \Gamma \to \{ \pm 1 \}$ is a character of $\Gamma$ and the axioms are in fact equivalent to the following condition on $h$ acting on $\Gamma$-orbits:
$h(\gamma \ox)= \epsilon(\gamma) h(\ox).$
When $n=2,$ the $\Gamma$-orbits are never injective, and this is indeed the key to the uniqueness part of the theorem. For $n \ge 3,$ however, one can see that the $\Gamma$-action is generically free, that is the map 
$\gamma \mapsto \gamma \cdot \ox$ is a bijection. 
\end{remark}
\section{The General Case}
\par 
Our approach gives a deterministic answer to Siegel's paradox when there are two equally likely exchange rates in the future, namely that the geometric mean is the {\it only} possible aggregator subject to Symmetry, Redenomination and Reciprocity axioms. 
One might think of other variations of this problem as well, for example when there are more than two such possible future exchange rates $e_1, \dots, e_n.$ 
An aggregator then is a function $\A: E^n \to E,$ where $E^n$ denotes the $n$-fold Cartesian product of $E$ with itself. We formulate the following axioms:
\begin{enumerate}
\item[\bf{A1.}] {\it Symmetry:} For all $e_1, \dots, e_n \in E$, 
we have $\A(e_1, \dots, e_n)=\A(e_{ \sigma(1)}, \dots,  e_{\sigma(n)})$ 
for any permutation $\sigma$ of the set $\{ 1, 2, \dots, n \}.$
\item[\bf{A2.}]  {\it Scaling:} For all $e_1, \dots, e_n \in E$ and every $\lambda>0$, we have
$$\A( \lambda e_1, \dots,  \lambda e_n)= \lambda \cdot \A(e_1, \dots, e_n).$$
\item[\bf{A3.}]  {\it Reciprocity:} For all $e_1, \dots, e_n \in E$, we have 
$$\A(e_1^{-1}, \dots, e_n^{-1})=\A(e_1, \dots, e_n)^{-1}.$$
\end{enumerate}
Having known the result for two rates, it would be natural to assume that the $n$-aggregator in this case should again be the geometric mean  
$\A(e_1, \dots , e_n)=\sqrt[n]{e_1 \dots e_n}.$
This function definitely satisfies the generalized axioms we have proposed above and would be the aggregator of choice. However, as the following proposition shows, this is not the only possibility. Let us call a collection of functions {\it log-convex} if their {\it logarithms} constitute a convex set. Recall that a set is convex if for any two points $A_0$ and $A_1$ in the set and any $\alpha \in [0,1],$ the convex combination $(1-\alpha)A_0+\alpha A_1$ is also in the set. 
\begin{proposition}\label{infinite-agg}
For $n>2$, there are infinitely many aggregators $\A: E^n \to E$ satisfying axioms {\bf{A1}}--{\bf{A3}}. The collection of all such aggregators is $\log$-convex.
\end{proposition}
\begin{proof}
Let $e^{(1)} \le \dots \le e^{(n)}$ denote the order statistics of the sequence $e_1, \dots, e_n.$ In other words, $e^{(1)}$ is the minimum of $e_1, \dots, e_n$, $e^{(2)}$ is the second smallest term, etc. 
For $n \ge 3$, let $\A(e_1, \dots, e_n)$ denote the median of $e_1, \dots, e_n.$ When $n$ is odd, this is defined to be $ e^{(\frac{n+1}{2})},$ and when $n$ is even, it is defined to be $\sqrt{e^{(\frac{n}{2})}  e^{(\frac{n}{2}+1)} }.$
We are going to show that the median is an aggregator satisfying {\bf{A1}}--{\bf{A3}}. 
Notice that the order statistics are invariant under permutations and scale with $\lambda>0$ when the original sequence is scaled. 
It remains to prove {\bf{A3}}. Note that the order statistics for $1/e_1, \dots, 1/e_n$ are simply $1/e^{(n)} \le \dots \le 1/e^{(1)}.$ 
This proves the claim. 
\par
One can verify that if $\A_0$ and $\A_1$ are two aggregators satisfying {\bf{A1}}--{\bf{A3}}, and $\alpha \in [0,1],$ then $\A_{\alpha}= \A_0^{1- \alpha} \A_1^{\alpha} $ is also an aggregator satisfying {\bf{A1}}--{\bf{A3}}. 
One can also check that these aggregators are all different if $\A_0$ and $\A_1$ are different (in this case, the geometric mean and the median). 
This proves that there are infinitely many aggregators. 
Moreover, since $\log\A_{\alpha}$ is a convex combination of $\log\A_0$ and $\log\A_1,$ 
the collection of aggregators satisfying {\bf{A1}}--{\bf{A3}} is log-convex. 
\end{proof}
\par
We can in fact go further and characterize all aggregators satisfying {\bf{A1}}--{\bf{A3}} in any dimension. 
As before, consider the order statistics $e^{(1)} \le \dots \le e^{(n)}$ 
and form the consecutive ratios $e^{(n)}/e^{(n-1)}, e^{(n-1)}/e^{(n-2)}, \dots, e^{(2)}/e^{(1)}.$ 
We have the following general characterization theorem. 
\begin{theorem}
\label{mainn}
Any aggregator $\A: E^n \to E$ satisfying {\bf{A1}}--{\bf{A3}} is of the form 
\begin{equation}\label{aggregatorsn}
\A(e_1, \dots, e_n)=(e_1\dots e_n)^{1/n} \beta(e^{(n)}/e^{(n-1)}, \dots, e^{(2)}/e^{(1)}),
\end{equation}
where $\beta: (\RR^{\ge 1})^{n-1} \to \RR^{+}$ is a function satisfying 
$\beta(u_1, \dots, u_{n-1}) \beta(u_{n-1}, \dots, u_1)=1.$ 
Conversely, for any such function $\beta,$ the aggregator $\A$ defined by \eqref{aggregatorsn} satisfies {\bf{A1}}--{\bf{A3}}. 
\end{theorem} 
For simplicity, let us call a function $\beta$ satisfying $\beta(u_1, \dots, u_{n-1}) \beta(u_{n-1}, \dots, u_1)=1$ 
a {\it reciprocity} function. Note that the constant function $\beta=1$ is always a reciprocity function, and using it in formula \eqref{aggregatorsn} produces the familiar geometric mean aggregator. For a non-trivial example, let $n=3$ and take $\beta(u_1,u_2)=(u_2/u_1)^{1/3}$ to obtain the median as another aggregator, as we have verified before. 
\begin{proof}
Let us prove the converse first. For any aggregator $\A$ defined as above, we can see that it satisfies {\bf{A1}}--{\bf{A3}} by verifying how the two factors on the right-hand side of the formula transform under symmetry, scaling and reciprocals. We know how $(e_1\dots e_n)^{1/n}$ transforms, because it already satisfies {\bf{A1}}--{\bf{A3}}. It is also pretty clear that the value of $ \beta(e^{(n)}/e^{(n-1)}, \dots, e^{(2)}/e^{(1)})$ is invariant under a scaling of $e_1, \dots, e_n,$ as well as their permutations, 
so all we need to check is how it transforms if $e_1, \dots, e_n$ are replaced by their reciprocals ${e_1}^{-1}, \dots, {e_n}^{-1};$ this will just reverse the order of the variables $e^{(n)}/e^{(n-1)}, \dots, e^{(2)}/e^{(1)},$ so we get the reciprocal value for $\beta$, since it is a reciprocity function. 
\par
This argument meanwhile implies that if formula \eqref{aggregatorsn} holds, then the aggregator $\A$ automatically satisfies {\bf{A1}} and {\bf{A2}}, and satisfies {\bf{A3}} exactly when the function $\beta$ is a reciprocity function, since the $n$-tuple $(e^{(n)}/e^{(n-1)}, \dots, e^{(2)}/e^{(1)})$ can assume any arbitrary value in $(\RR^{\ge 1})^{n-1}$ by an appropriate choice of $e_1, \dots, e_n.$
\par
Now, let us prove the other direction of the theorem. Since $\A$ is invariant under permutations, we can assume without loss of generality that $e_1 \le \dots \le e_n.$ We now need to establish the existence of a reciprocity function $\beta: (\RR^{+})^{n-1} \to \RR^{+}$ such that 
$$\A(e_1, \dots, e_n)=(e_1\dots e_n)^{1/n} \beta(e_n/e_{n-1}, \dots, e_2/e_1).$$
As in the case of $n=2,$ we set 
$$h(x_1, \dots, x_n)=\log(\A(e^{x_1}, \dots, e^{x_n})) - \frac{x_1+\dots+x_n}{n}.$$
One can again see that as a result of {\bf{A1}}--{\bf{A3}}, $h$ is symmetric, translation-invariant, and satisfies
$h(-x_1,\dots,-x_n)= -h(x_1,\dots,x_n).$ 
For $u_1, \dots, u_{n-1} \in \RR^+,$ set 
$$\beta(u_{n-1}, u_{n-2}, \dots, u_1)=\exp(h(0, \log u_1, \log(u_1 u_2), \dots, \log(u_1 \dots u_{n-1})).$$ 
\par
The claim now follows by combining these formulas for $e_1=e^{x_1}, \dots, e_n=e^{x_n}.$ 
We just need to keep in mind that, by translation invariance, 
$$h(x_1, \dots, x_n)=h(0, x_2-x_1, \dots, x_n-x_1).$$
\par
As we noted earlier, since formula \eqref{aggregatorsn} holds and $\A$ satisfies the reciprocity axiom {\bf{A3}}, $\beta$ has to be a reciprocity function. However, to see this explicitly, write 
\begin{equation*}
\begin{split}
h(0, \log u_1, \log(u_1 u_2), \dots, \log(u_1 \dots u_{n-1}) & =-h(0, -\log u_1, \dots, -\log(u_1 \dots u_{n-1})) \\
& =-h(\log(u_1 \dots u_{n-1}), \log((u_2 \dots u_{n-1})), \dots, 0) \\
& =-h(0, \log u_{n-1}, \dots, \log(u_1 \dots u_{n-1})), 
\end{split}
\end{equation*}
where we have translated the variables by $\log(u_1 \dots u_{n-1})$ in the second step and used symmetry in the last. Exponentiating the first and the last terms here produce $\beta(u_{n-1}, u_{n-2}, \dots, u_1)$ and $\beta(u_1, u_2, \dots, u_{n-1})^{-1}.$
\end{proof}
\section{Discussion}
\par
We have been able to classify all aggregators satisfying {\bf{A1}}--{\bf{A3}} in terms of the geometric mean and a reciprocity function $\beta,$ as in theorem \ref{mainn}. For $n=2,$ this function $\beta$ is identical to 1, so we obtain a unique aggregator in this case, but there are plenty of options in other cases. It would be interesting to find out an interpretation of the reciprocity function $\beta$ in economic terms. This could help our understanding of aggregators and may even point to some natural axioms or constraints 
that could narrow down the possibilities for an aggregator or even determine it uniquely. In all cases, however, the geometric mean stands out as the trivial aggregator of choice. 
\par
Another direction in which the problem could be generalized would be when the future exchange rates $e_1, \dots, e_n$ are {\it not} necessarily equally likely, but occur with probabilities $p_1, \dots, p_n,$ respectively. It would be fair to assume that  the following weighted geometric mean would be the natural aggregator in this case. 
\begin{equation}\label{weighted}
\A(e_1, \dots, e_n)= e_1^{p_1} e_2^{p_2} \dots e_n^{p_n}.
\end{equation}
One could in fact argue for this aggregator, at least when $p_1, \dots, p_n$ are all rational numbers, but should take extra care in interpreting the Symmetry axiom {\bf{A1}}. To demonstrate this case, let us assume that there are only two possible future exchange rates $e_1$ and $e_2,$ occurring with probabilities $m/(m+n)$ and $n/(m+n),$ respectively. We can recast this situation as the case of $m+n$ equally likely future exchange rates, where $m$ of those possibilities are $e_1$ and the rest are $e_2.$ Then we can apply theorem \ref{mainn} to find all the aggregators in this case, in particular our distinguished geometric mean, which would be a weighted one, as suggested in equation \eqref{weighted}. This argument could be adapted to extend to any number of future exchange rates with rational probabilities, but runs into difficulty when the probabilities $p_1, \dots, p_n,$ contain irrational numbers. Fixing this issue may require an appropriate modification of the Symmetry axiom or formulation of an additional {\it Continuity} axiom to extend the results to the general case and could open up new avenues for investigation.
\section{Acknowledgements}
\par
We wish to thank Hamed Ghoddusi for introducing Siegel's paradox to us and commenting on an earlier draft of this paper. We would also like to thank Hazhir Rahmandad and Hassan Tehranian for their helpful comments on the paper which led to improvements in the text, as well as Ken Rogoff for his interest in our work.
%
%
%
\bibliographystyle{apa-good}
\bibliography{siegel} 
%
%
%
\end{document}